\documentclass[11pt]{article}

\usepackage{amssymb}
\usepackage{amsmath}
\usepackage{amsthm}
\usepackage{color}
\usepackage{graphicx}
\usepackage{hyperref}
\usepackage{fullpage}
\usepackage{graphicx}
\usepackage{mathtools}
\usepackage{microtype}
\usepackage{bbm}

\usepackage[linesnumbered,algoruled,titlenotnumbered,vlined]{algorithm2e}
\DontPrintSemicolon

\usepackage{algorithmic}

\usepackage{lmodern}
\usepackage[T1]{fontenc}

\newtheorem{theorem}{Theorem}
\newtheorem{lemma}[theorem]{Lemma}

\newtheorem{claim}[theorem]{Claim}

\newtheorem{corollary}[theorem]{Corollary}

\renewenvironment{proof}{\vspace{-0.05in}\noindent{\bf Proof:}}%
        {\hspace*{\fill}$\Box$\par}
\newenvironment{proofof}[1]{\smallskip\noindent{\bf Proof of #1:}}%
        {\hspace*{\fill}$\Box$\par}
\newenvironment{proofsketch}{\vspace{-0.1in}\noindent{\bf Proof sketch:}}%
        {\hspace*{\fill}$\Box$\par}

\newcommand{\eps}{\epsilon}
\newcommand{\Ex}{\mathbb{E}}

\newcommand{\R}{{\mathbb{R}}}

\newcommand{\one}{\mathbf{1}}
\newcommand{\zero}{\mathbf{0}}

\newcommand{\bc}{{\bf c}}
\newcommand{\bd}{{\bf d}}
\newcommand{\bx}{{\bf x}}
\newcommand{\by}{{\bf y}}
\newcommand{\bz}{{\bf z}}
\newcommand{\bw}{{\bf w}}
\newcommand{\bu}{{\bf u}}
\newcommand{\bv}{{\bf v}}
\newcommand{\bp}{{\bf p}}

\newcommand{\bo}{{\bf o}}

\def\OPT{\mathrm{OPT}}

\DeclareMathOperator{\argmax}{argmax}

\def\({\left(}
\def\){\right)}

\def\etal{\emph{et al.}\xspace}

\begin{document}
\title{Constrained Submodular Maximization: Beyond $1/e$}
\author{Alina Ene\thanks{Department of Computer Science, Boston University, {\tt aene@bu.edu}. This work was done in part while the author was with the Computer Science department at the University of Warwick and a visitor at the Toyota Technological Institute at Chicago.}
\and
Huy L. Nguy\~{\^{e}}n\thanks{College of Computer and Information Science, Northeastern University, {\tt hlnguyen@cs.princeton.edu}. This work was done in part while the author was with the Toyota Technological Institute at Chicago.}
}
\begin{titlepage}

\maketitle

\begin{abstract}
	In this work, we present a new algorithm for maximizing a non-monotone submodular function subject to a general constraint. Our algorithm finds an approximate fractional solution for maximizing the multilinear extension of the function over a down-closed polytope. The approximation guarantee is $0.372$ and it is the first improvement over the $1/e$ approximation achieved by the unified Continuous Greedy algorithm [Feldman \etal, FOCS 2011]. 
\end{abstract}

\end{titlepage}

\newpage
\section{Introduction}

A set function $f:2^{V}\rightarrow \R$ is {\em submodular} if for every $A, B \subseteq V$, we have $f(A)+f(B) \ge f(A\cup B) + f(A\cap B)$. Submodular functions naturally arise in a variety of contexts, both in theory and practice. From a theoretical perspective, submodular functions provide a common abstraction of cut functions of graphs and digraphs, Shannon entropy, weighted coverage functions, and log-determinants. From a practical perspective, submodular functions are used in a wide range of application domains from machine learning to economics. In machine learning, it is used for document summarization~\cite{LinB10}, sensor placement~\cite{KrauseSG08}, exemplar clustering~\cite{GomesK10}, potential functions for image segmentation~\cite{JegelkaB11}, etc. In an economics context, it can be used to model market expansion~\cite{DughmiRS12}, influence in social networks~\cite{Kempe2003}, etc. The core mathematical problem underpinning many of these applications is the meta problem of maximizing a submodular objective function subject to some constraints, i.e., $\max_{S\in \mathcal{S}} f(S)$ where $\mathcal{S}$ is a down-closed family of sets\footnote{A family of sets $\mathcal{S}$ is down-closed if for all $A\subset B$, if $B\in \mathcal{S}$ then $A\in \mathcal{S}$.}.

A common approach to this problem is a two-step framework based on the multilinear extension $F$ of $f$, a continuous function that extends $f$ to the domain $[0,1]^V$. The program first (1) maximizes $F(\bx)$ subject to $\bx\in C$ where $C$ is a convex relaxation of $\mathcal{S}$, and then (2) rounds $\bx$ to an integral solution in $\mathcal{S}$. This paradigm has been very successful and it has led to the current best approximation algorithms for a wide variety of constraints including cardinality constraints, knapsack constraints, matroid constraints, etc. The contention resolution scheme framework of Chekuri \etal~\cite{ChekuriVZ14} gives a generic way to perform step 2 of the program. Thus, many recent works~\cite{FeldmanNS11,ChekuriJV15} focus on improving step 1, as it immediately leads to improved algorithms when combined with the known rounding schemes. Feldman \etal~\cite{FeldmanNS11} give a beautiful generalization of the continuous greedy algorithm~\cite{Vondrak2008}, achieving an $1/e$ approximation for step 1 for any down-closed and solvable polytope\footnote{A polytope $C$ is solvable if there is an oracle for optimizing linear functions over $C$, i.e., for solving $\max_{\bc\in C} \langle \bv, \bc\rangle$ for any vector $\bv$.} $C$. While it is known that the continuous greedy algorithm is optimal when we restrict to monotone functions\footnote{A function $f: 2^{V}\rightarrow \R$ is monotone if $f(A) \leq f(B)$ for all $A \subseteq B$.}, it is not known if the $1/e$ approximation of the extension by~\cite{FeldmanNS11} is optimal. Recently, the work of Buchbinder \etal~\cite{Buchbinder2014a} shows that, for the special case of a \emph{cardinality} constraint, it is possible to beat $1/e$. However, this result still leaves open the possibility that an $1/e$ approximation is best possible for a harder constraint. This possibility is consistent with our current knowledge: the best known hardness for a cardinality constraint is 0.491 while the best known hardness for a matroid constraint is 0.478~\cite{Gharan2011}, suggesting that the matroid constraint may be strictly harder to approximate. In this paper, we rule out this possibility and show that it is possible to go beyond the $1/e$ barrier in the same generic setting as considered by~\cite{FeldmanNS11}.

\begin{theorem}
  Let $f$ be a non-negative, non-monotone, submodular function $f$. Let $F$ be the multilinear extension of $f$. Let $C$ be a down-closed and solvable polytope. There is an efficient algorithm that constructs a solution $\bx \in C$ such that $F(\bx) \geq 0.372 \cdot F(\OPT)$, where $\OPT$ is an optimal integral solution to the problem $\max_{\bx \in C \cap \{0, 1\}^V} F(\bx)$.
\end{theorem}

In order to keep the analysis as simple as possible, we have not tried to optimize the constant in the theorem above. We believe a better constant can be obtained using the techniques in this paper, and we hope that our work will lead to improved approximation guarantees for the problem.

Using known rounding techniques, we obtain improved approximation guarantees for several classes of constraints. In particular, we obtain a $0.372 - o(1)$ approximation for submodular maximization subject to a matroid constraint and a $0.372 - \eps$ approximation for a constant number of knapsack constraints, which improve over the $1/e - o(1)$ and $1/e - \eps$ approximations \cite{FeldmanNS11}. 

\subsection{Our techniques}
Our starting point is the unified continuous greedy algorithm~\cite{FeldmanNS11} achieving the approximation factor $1/e$. The algorithm grows a solution $\bx$ over the time interval $[0,1]$, improving $F(\bx)$ in each time step by an amount proportional to $(1-\|\bx\|_\infty) F(\OPT)$. As $\bx$ increases over time, the gain starts out large and decreases to 0 at the end of the process. The change in $\bx$ in each time step is proportional\footnote{We use $\bx \circ \by$ to denote the vector whose $i$-th coordinate is $\bx_i \cdot \by_i$.} to $\bv = \argmax_{\bc\in C} \langle \nabla F(\bx)\circ (\one-\bx), \bc\rangle$. Notice that we can improve the gain ($(1-\|\bx\|_{\infty}) F(\OPT)$) by slowing down the growth of $\|\bx\|_{\infty}$. Thus in our algorithm, we add a new constraint $\|\bc\|_{\infty}\le \alpha$ and have $\bv = \argmax_{\bc\in C, \|\bc\|_{\infty} \le \alpha} \langle \nabla F(\bx)\circ (\one-\bx), \bc\rangle$. How does this change affect the performance of the algorithm and in particular, the quantity $\langle \nabla F(\bx)\circ (\one-\bx), \bv\rangle$? Intuitively, if there is a second solution other than $\OPT$ with value close to $\OPT$, we can pick $\bv$ to be a mixture of $\OPT$ and this solution and have a new solution comparable to $\OPT$ but with lower $\ell_\infty$ norm. Thus, if this is the case, we do not lose very much in $\langle \nabla F(\bx)\circ (\one-\bx), \bv\rangle$ while simultaneously increasing the gain ($1-\|\bx\|_{\infty})F(\OPT)$. On the other hand, if there is no such solution, the crucial insight is that $\bv$ must be \emph{well correlated} with $\OPT$. Thus we can identify a good fraction of $\OPT$ by searching for $\argmax_{\bx\le \bv} F(\bx)$. (Here we crucially use the fact that $C$ is a down-closed polytope.) This problem turns out to be not very different from the \emph{unconstrained} setting and we can use a variant of the double greedy algorithm of Buchbinder \etal~\cite{Buchbinder2012} to find a good solution.

The above description is an intuitive but simplified overview of the algorithm. Nevertheless, each routine mentioned above corresponds to a part of the algorithm described in Section~\ref{sec:algorithm}. The main technical difficulty is in formalizing the statement ``$\bv$ must be well correlated with $\OPT$'' and the subsequent identification of (a large part of) $\OPT$. We describe how to overcome these difficulties in Section~\ref{sec:z-value}.

\smallskip\noindent
{\bf Comparison with the work of \cite{Buchbinder2014a} for a cardinality constraint.}
The idea of adding a new constraint $\|\bc \|_{\infty} \le \alpha$ is inspired by a recent algorithm by~\cite{Buchbinder2014a} for the cardinality constraint. In their setting, the goal is to maximize $F(\bx)$ subject to a constraint $\|\bx\|_1 \le k$. In the $i$th step, their algorithm picks $2(k-i)$ elements with maximum marginal gain and randomly adds one of them to the current solution. This can be viewed as an analog of adding a constraint $\|\bc\|_\infty \le k/(2k-2i)$ when the time is in the range $[(i-1)/k, i/k]$. This more sophisticated use of varying thresholds depending on the time is allowed in our framework but in the simple solution presented here, we just use a fixed threshold throughout.

When the marginal gain ($\langle \nabla F(\bx)\circ (1-\bx), \bv \rangle$ in our case) is small, they also use a variant of the double greedy algorithm to finish the solution. From the $2(k-i)$ elements that are picked in the last iteration, their algorithm picks $k-i$ elements using a sophisticated variant of double Greedy, and add them to the $i$ elements it picked in the previous iterations. Unfortunately this step crucially uses the structure of the solution when the constraint is a cardinality constraint. From the point of view of a continuous Greedy algorithm, this is analogous to doubling a fractional solution and then selecting half of the coordinates. It is not clear what the analog should be in the general case. Instead, we compute a different solution $\bp=\argmax_{\bc\in C} \langle \nabla F(\bx)\circ (1-\bx), \bc \rangle$ (dropping the dampening constraint) and use standard double greedy to find $\bz=\argmax_{\bc\le \bp} F(\bc)$. While $\bv$ and $\bp$ could potentially be very different, we manage to connect $\langle \nabla F(\bx)\circ (1-\bx), \bv \rangle$ and $\langle \nabla F(\bx)\circ (1-\bx), \bp \rangle$, which is enough to prove the existence of $\bz\le \bp$ with large $F(\bz)$.

In summary, in this work, we introduce several novel insights that lead to a much more general algorithm and analysis. We believe that our algorithm and analysis are conceptually much simpler and cleaner, and we hope that our techniques will lead to further improvements for the problem.

\subsection{Related work}

Submodular maximization problems are well-studied and it is hard to do justice to all previous literature on this subject. Starting with the seminal work Nemhauser et al.~\cite{Nemhauser1978}, the classical approaches to these problems are largely combinatorial and based on greedy and local search algorithms. Over the years, with increasing sophistication, this direction has led to many tight results such as the algorithm of~\cite{Sviridenko2004} for a knapsack constraint, and the current best results for constraints such as multiple matroid constraints~\cite{Lee2010a}. In the last few years, another approach emerged~\cite{Calinescu2011} that follows the popular paradigm in approximation algorithms of optimizing a continuous relaxation and rounding the resulting fractional solution. A key difficulty that separates the submodular setting from the classical setting is that even finding a \emph{fractional solution} may be quite challenging, and in particular it is $\mathrm{\bf NP}$-hard to solve the continuous relaxation for maximizing submodular functions that is based on the multilinear extension. Thus, a line of work has been developed to \emph{approximately} optimize this relaxation~\cite{Calinescu2011,Lee2009,Gharan2011,FeldmanNS11} culminating in the work~\cite{FeldmanNS11}, which we extend here.

\section{Preliminaries}
Here we consider the problem of maximizing a submodular function $f$ subject to a downward closed convex constraint $\bx\in C$.

We use the following notation. Let $n = |V|$. We write $\bx \le \by$ if $\bx_i \le \by_i$ for all coordinates $i \in [n]$. Let $\bx \circ \by$ denote the vector whose $i$-th coordinate is $\bx_i \cdot \by_i$. Let $\bx \vee \by$ (resp. $\bx \wedge \by$) be the vector whose $i$-th coordinate is $\max\{\bx_i, \by_i\}$ (resp. $\min\{\bx_i, \by_i\}$). Let $\one_S \in \{0, 1\}^V$ denote the indicator vector of $S \subseteq V$, i.e., the vector that has a $1$ in entry $i$ if and only if $i \in S$. 

Let $F: [0, 1]^n \rightarrow \R_+$ denote the multilinear extension of $f$:
	\[ F(\bx) = \Ex[f(R(\bx))] = \sum_{S \subseteq V} f(S) \prod_{i \in S} \bx_i \prod_{i \in V \setminus S} (1 - \bx_i),\]
where $R(\bx)$ is a random subset of $V$ where each $i \in V$ is included independently at random with probability $\bx_i$.
We use $\nabla F$ to denote the gradient of $F$ and we use ${\partial F \over \partial \bx_i}$ to denote the $i$-th coordinate of the gradient of $F$.

The multilinear extension has the following well-known properties, see for instance \cite{Vondrak2008}.

\begin{claim}
\label{fact1}
	$\frac{\partial F}{\partial x_i}(\bx) = F(\bx \vee \one_{\{i\}}) - F(\bx \wedge \one_{V\setminus \{i\}})$.
\end{claim}
\begin{proof}
  Note that
    \[ F(\bx) = \bx_i \cdot F(\bx \vee \one_{\{i\}}) + (1 - \bx_i) \cdot F(\bx \wedge \one_{V\setminus \{i\}}).\]
  Thus if we take the partial derivative with respect to $\bx_i$, we obtain the claim. 
\end{proof}

\begin{claim}
\label{fact2}
	If $\bx\le \by$ then $\nabla F(\bx) \ge \nabla F(\by)$.
\end{claim}
\begin{proof}
  Fix a coordinate $i$. Since $\bx \leq \by$, submodularity implies that
  \[ F(\bx \vee \one_{\{i\}}) - F(\bx \wedge \one_{V\setminus \{i\}}) \geq F(\by \vee \one_{\{i\}}) - F(\by \wedge \one_{V\setminus \{i\}}). \]
  By Claim~\ref{fact1}, the left-hand side of the inequality above is ${\partial F \over \partial x_i}(\bx)$ and the right-hand side is ${\partial F \over \partial x_i}(\by)$.
\end{proof}

\begin{claim}
\label{fact3}
	$F$ is concave along any line of direction $\bd \geq \zero$. That is, for any vector $\bx$, the function $\phi: \R \rightarrow \R$ such that $\phi(t) = F(\bx + t \cdot \bd)$ for every $t \in \R$ for which $\bx + t \cdot \bd \in [0, 1]^V$ is concave. 
\end{claim}
\begin{proofsketch}
  By submodularity, we can verify that $\phi''(t) \leq 0$ and thus $\phi$ is concave. 
\end{proofsketch}

\section{Double Greedy algorithm for a box constraint}
\label{sec:double-greedy-box}

In this section, we describe an algorithm for maximizing the multilinear extension subject to a box constraint: given $\bu$ and $\bv$, find $\max_{\bu\le \bx \le \bv} F(\bx)$. The algorithm is similar to the Double Greedy algorithm of \cite{Buchbinder2012} and it is given in Figure~\ref{fig:double-greedy-box}.

\begin{figure}[t!]
\centering
\begin{minipage}{0.7\linewidth}
\begin{algorithm}[H]
\label{alg:double-greedy-box}
\caption{Algorithm for a box constraint}
	$\bu^{(0)} = \bu$, $\bv^{(0)} = \bv$ \;
	\For{$i\in [n]$}{
  	$a_i \gets (\bv_i - \bu_i) \left(F\left(\bu^{(i-1)} \vee \one_{\{i\}} \right) - F\left(\bu^{(i-1)} \wedge \one_{V\setminus\{i\}} \right) \right)$ \;
		$b_i \gets (\bv_i - \bu_i) \left(F\left(\bv^{(i-1)} \wedge \one_{V\setminus\{i\}} \right) - F\left(\bv^{(i-1)} \vee \one_{\{i\}} \right)\right)$ \;
  	$a_i' = \max(a_i, 0), b_i' = \max(b_i, 0)$ \;
  	\If{$a'_i + b'_i \neq 0$}{
			$\bu^{(i)} \gets \bu^{(i-1)} + \frac{a_i'}{a_i'+b_i'} \cdot (\bv_i - \bu_i) \cdot \one_{\{i\}}$ \;
  		$\bv^{(i)} \gets \bv^{(i-1)} - \frac{b_i'}{a_i'+b_i'} \cdot (\bv_i - \bu_i) \cdot \one_{\{i\}}$ \;
		}\Else{
     \tcp{Can update $\bu^{(i)}_i$ and $\bv^{(i)}_i$ to any common value}
     \tcp{We set $\bu^{(i)}_i = \bv^{(i)}_i = \bv^{(i - 1)}_i$}
		  $\bu^{(i)} \gets \bu^{(i - 1)} + (\bv_i - \bu_i) \cdot \one_{\{i\}}$ \;
      $\bv^{(i)} \gets \bv^{(i - 1)}$ \;
    }
	}
	{\bf return} $\bu^{(n)}$ \;
\end{algorithm}
\end{minipage}
\caption{Double Greedy algorithm for a box constraint $\{\bx \colon \bu \leq \bx \leq \bv\}$.}
\label{fig:double-greedy-box}
\end{figure}

The proof of the following lemma is similar to the analysis of the Double Greedy algorithm and it can be found in Appendix~\ref{app:double-greedy-box}.

\begin{lemma}\label{lem:box-constraint}
	Algorithm~\ref{alg:double-greedy-box} finds a solution $\bx$ to the problem $\max_{\bu \leq \bx \leq \bv} F(\bx)$ such that
		\[ F(\bx) \ge {1 \over 2} \cdot F(\OPT) + {1 \over 4} \cdot F(\bu) + {1 \over 4} \cdot F(\bv), \]
	where $\OPT$ is an optimal solution to the problem $\max_{\bu \leq \bx \leq \bv} F(\bx)$. 
\end{lemma}

Alternatively, we can reduce the problem with a box constraint to the unconstrained problem by defining a suitable submodular function $g(S) = F(\bu + (\bv - \bu) \circ \one_{S})$. We can show that this reduction is correct using the same argument given in the appendix.

\section{The algorithm}
\label{sec:algorithm}

\begin{figure}[t!]
\centering
\begin{minipage}{0.7\linewidth}
\begin{algorithm}[H]
\label{alg:main}
\caption{Algorithm for a general constraint}
Initialize $\bx^* = \zero$ \;
\For{$\theta \in [0, 1]$}{
  \tcp{Dampened Continuous Greedy stage}
	Initialize $\bx^{(0)} = \zero$ \;
	\For{$t \in [0, \theta]$}{
	$\bv^{(t)} = \argmax_{\bc \in C, \textcolor{red}{\|\bc\|_{\infty} \le \alpha}} \langle \nabla F(\bx^{(t)}) \circ (\one - \bx^{(t)}), \bc \rangle$ \;
  Update $\bx^{(t)}$ according to ${d \bx^{(t)} \over dt} = \bv^{(t)} \circ (\one - \bx^{(t)})$ \;
	}
  \;
  \tcp{Standard Continuous Greedy stage}
  Initialize $\by^{(\theta)} = \bx^{(\theta)}$ \;
  \For{$t \in (\theta, 1]$}{ 
    $\bv^{(t)} = \argmax_{\bc \in C} \langle \nabla F(\by^{(t)}) \circ (\one - \by^{(t)}), \bc \rangle$ \;
    Update $\by^{(t)}$ according to ${d \by^{(t)} \over dt} = \bv^{(t)} \circ (\one - \by^{(t)})$ \;
  }
  \;
  \tcp{Double Greedy stage}
  $\bp = \argmax_{\bc \in C} \langle \nabla F(\bx^{(\theta)}) \circ (\one - \bx^{(\theta)}), \bc \rangle$ \;
  Find $\bz$ approximating $\argmax_{\bc \le \bp} F(\bc)$ using the Double Greedy algorithm (Algorithm~\ref{alg:double-greedy-box})\;
  \;
  \tcp{Update the best solution}
  \If{$F(\by^{(1)}) > F(\bx^*)$}{
		$\bx^* = \by^{(1)}$ \;
  }
	\If{$F(\bz) > F(\bx^*)$}{
		$\bx^* = \bz$ \;
	}
}
Return $\bx^*$ \;
\end{algorithm}
\end{minipage}
\caption{Continuous algorithm for a general constraint}
\end{figure}

In this section, we describe our main algorithm for the problem $\max_{\bx \in C} F(\bx)$. We first give a continuous version of our algorithm; in order to efficiently implement the algorithm, we discretize it using a standard approach, and we give the details in Appendix~\ref{app:discretized}.

The algorithm picks the best out of two solutions. The first solution is constructed by running the Continuous Greedy algorithm \emph{with an additional dampening constraint\footnote{The dampening constraint imposes an $\ell_{\infty}$ constraint in line $5$ of Algorithm~\ref{alg:main}.}} as long as the marginal gain remains high despite the dampening constraint, and then finishing the solution via the standard Continuous Greedy algorithm for the remaining time. The second solution is constructed by running Double Greedy exactly when the marginal gain becomes low because of the dampening constraint, which must happen early if the first solution is not good. Since we do not know precisely when the marginal gain becomes low, the algorithm tries all possible values via the outer for loop on line 2.

\section{The analysis}
\label{sec:analysis}

In this section, we analyze the algorithm and show that it achieves an approximation greater than $0.372$. We remark that the analysis in Subsection~\ref{sec:y-value} is fairly standard, and the crux of the analysis is in Subsection~\ref{sec:z-value}.

In the following, we fix an iteration $\theta$ of the outer loop and we analyze the solutions constructed in that iteration. In the remainder of this section, all of the vectors $\bx^{(\cdotp)}$, $\by^{(\cdotp)}$, etc. refer to the vectors during iteration $\theta$.

\subsection{Analysis of the solution $\by^{(1)}$}
\label{sec:y-value}

In this section, we analyze the solution $\by^{(1)}$ constructed in any given iteration $\theta \in [0, 1]$.

\begin{theorem} \label{thm:y-value}
   Let $\theta \in [0, 1]$ and let $\by^{(1)}$ be the solution constructed by Algorithm~\ref{alg:main} in iteration $\theta$ of the outer loop. We have
		\[ F(\by^{(1)}) \geq  {1 \over e} \left( e^{\theta} F(\bx^{(\theta)}) + (1 - \theta) e^{(1 - \alpha) \theta} F(\OPT)\right).\]
\end{theorem}

We devote the rest of this section to the proof of Theorem~\ref{thm:y-value}. In the remainder of this section, all of the vectors $\bx^{(\cdot)}$, $\by^{(\cdot)}$, etc. refer to the vectors during iteration $\theta$.

We start by upper bounding the $\ell_{\infty}$ norm of $\bx^{(t)}$ and $\by^{(t)}$.

\begin{lemma} \label{lem:max-norm}
	Consider the following process. Let time run from $t = t_0 \geq 0$ to $t = t_1 \leq 1$. Let $\bu^{(t)}$ be a vector updated according to
		\[ {d \bu^{(t)} \over dt} = \bv^{(t)} \circ (\one - \bu^{(t)}),\]
	where $\bv^{(t)}$ is a vector such that $\|\bv^{(t)}\|_{\infty} \leq \delta$. Then $\|\bu^{(t)}\|_{\infty} \leq 1 + (\|\bu^{(t_0)}\|_{\infty} - 1) e^{- \delta (t - t_0)}$ for each $t \in [t_0, t_1]$.
\end{lemma}
\begin{proof}
	The $i$-th coordinate $\bu_i^{(t)}$ of $\bu^{(t)}$ is updated according to
		\[ {d \bu_i^{(t)} \over dt} = \bv_i^{(t)} (1 - \bu_i^{(t)}) \le \delta (1 - \bu_i^{(t)}). \]
	By solving the differential inequality above, we obtain $\bu_i^{(t)} \leq 1 + C \cdot e^{- \delta t}$, where $C$ is a constant. Using the initial condition $1 + C \cdot e^{- \delta t_0} = \bu_i^{(t_0)}$, we obtain $C = (\bu_i^{(t_0)} - 1) e^{\delta t_0}$ and thus $\bu_i^{(t)} \leq 1 + (\bu_i^{(t_0)} - 1) e^{- \delta (t - t_0)}$.
\end{proof}

\begin{corollary} \label{cor:max-x}
  For every $t \in [0, \theta]$, $\|\bx^{(t)}\|_{\infty} \leq 1 - e^{- \alpha t}$.
\end{corollary}
\begin{proof}
  The vector $\bx^{(t)}$ starts at $\zero$ and it is updated according to the update rule in line $6$, where $\bv^{(t)}$ is a vector of $\ell_{\infty}$ norm at most $\alpha$. Thus it follows from Lemma~\ref{lem:max-norm} that $\|\bx^{(t)} \|_{\infty} \leq 1 - e^{- \alpha t}$.
\end{proof}

\begin{corollary} \label{cor:max-y}
   For every $t \in [\theta, 1]$, $\| \by^{(t)} \|_{\infty} \leq 1 - e^{(1 - \alpha) \theta - t}$.
\end{corollary}
\begin{proof}
  The vector $\by^{(t)}$ starts at $\by^{(\theta)} = \bx^{(\theta)}$ and it is updated according to the update rule in line 11, where $\bv^{(t)}$ is a vector of $\ell_{\infty}$ norm at most $1$. Thus it follows from Lemma~\ref{lem:max-norm} and the upper bound on $\|\bx^{(\theta)}\|_{\infty}$ given by Corollary~\ref{cor:max-x} that
    \[ \|\by^{(t)} \|_{\infty} \leq 1 + (\|\bx^{(\theta)}\|_{\infty} - 1) e^{- (t - \theta)} \leq 1 - e^{(1 - \alpha) \theta - t}.\]
\end{proof}

We will also need the following lemma that was shown in \cite{FeldmanNS11}.

\begin{lemma}[{\cite[Lemma III.5]{FeldmanNS11}}] \label{lem:x-or-opt}
  Let $\bx \in [0, 1]^n$ and let $S \subseteq V$. We have
    \[ F(\bx \vee \one_{S}) \geq (1 - \|\bx\|_{\infty}) f(S). \]
\end{lemma}

\begin{proofof}{Theorem~\ref{thm:y-value}}
	Using the chain rule, for every $t \in [\theta, 1]$, we have
	\begin{align*}
		{d F(\by^{(t)}) \over dt} &= \left\langle \nabla F(\by^{(t)}),  {d\by^{(t)} \over dt} \right \rangle\\
		&= \left\langle \nabla F(\by^{(t)}), \bv^{(t)} \circ (\one - \by^{(t)})\right \rangle\\
		&= \left\langle \nabla F(\by^{(t)}) \circ (\one - \by^{(t)}), \bv^{(t)} \right \rangle\\
		&\geq \left\langle \nabla F(\by^{(t)}) \circ (\one - \by^{(t)}), \OPT \right \rangle\\
		&\geq F(\by^{(t)} \vee \OPT) - F(\by^{(t)})\\
		&\geq e^{(1 - \alpha) \theta - t} \cdot F(\OPT) - F(\by^{(t)}). &\quad \mbox{(By Corollary~\ref{cor:max-y} and Lemma~\ref{lem:x-or-opt})}
	\end{align*}
	By solving the differential inequality, we obtain
		\[ F(\by^{(t)}) \geq {1 \over e^t} \left( e^{\theta} F(\by^{(\theta)}) + (t - \theta) e^{(1 - \alpha)\theta} F(\OPT)\right),\]
	and thus
		\[ F(\by^{(1)}) \geq {1 \over e} \left( e^{\theta} F(\by^{(\theta)}) + (1 - \theta) e^{(1 - \alpha)\theta} F(\OPT)\right).\]
	The theorem now follows, since $\by^{(\theta)} = \bx^{(\theta)}$.
\end{proofof}

\subsection{Analysis of the solution $\bz$}
\label{sec:z-value}

In this section, we analyze the solution $\bz$ constructed using the Double Greedy algorithm, and this is the crux of our argument.

\begin{theorem} \label{thm:z-value}
	Let $\theta \in [0, 1]$ and let $\bz$ be the solution constructed by Algorithm~\ref{alg:main} in iteration $\theta$ of the outer loop. If $\alpha \geq 1/2$, we have
		\[ F(\bz) \geq {1 \over 2(1 - \alpha)} (e^{-\alpha \theta} F(\OPT) - F(\bx^{(\theta)}) - \langle \nabla F(\bx^{(\theta)}) \circ (\one - \bx^{(\theta)}), \bv^{(\theta)} \rangle).\] 
\end{theorem}

We first give some intuition for the approach. Consider the solution $\bx$ constructed by the \emph{dampened} Continuous Greedy algorithm. The rate of growth of $F(\bx)$ is given by the inner product $\langle \nabla F(\bx) \circ (\one - \bx), \bv \rangle$ and thus the crux of the analysis is to understand how this inner product evolves over time. The inner product $\langle \nabla F(\bx) \circ (\one - \bx), \bv \rangle$ is always at least $\langle \nabla F(\bx) \circ (\one - \bx), \alpha \OPT \rangle$ and intuitively we should gain proportional to the difference between the two. If this difference is very small, the key insight is that once we drop the dampening constraint and compute the vector that maximizes $\langle \nabla F(\bx) \circ (\one - \bx), \bc \rangle$ over all feasible vectors $\bc$, we obtain a vector $\bp$ that is well-correlated with $\OPT$ in the sense that $\bp \wedge \OPT$ is a good solution. We formalize this intuition in the remainder of this section. 

In the remainder of this section, we analyze the solution
	\[ \widehat{\bz} \coloneqq (\one - \bx^{(\theta)}) \circ (\bp \wedge \OPT).\]
By submodularity, we have
	\[ F(\widehat{\bz}) - F(\mathbf{0}) \geq F(\widehat{\bz} + \bx^{(\theta)}) - F(\bx^{(\theta)}),\]
and thus it suffices to analyze $F(\widehat{\bz} + \bx^{(\theta)})$.

Note that $\bx^{(\theta)} \leq \widehat{\bz} + \bx^{(\theta)} \leq \bx^{(\theta)} + (\one - \bx^{(\theta)}) \circ \OPT = \bx^{(\theta)} \vee \OPT$. Thus Claim~\ref{fact3} and Claim~\ref{fact2} give
\begin{align*}
	F(\bx^{(\theta)} \vee \OPT) - F(\widehat{\bz} + \bx^{(\theta)}) &\leq \langle \nabla F(\widehat{\bz} + \bx^{(\theta)}), (\one - \bx^{(\theta)}) \circ ((\one - \bp) \wedge \OPT)\rangle &\quad \mbox{(By Claim~\ref{fact3})}\\
	&\leq \langle \nabla F(\bx^{(\theta)}), (\one - \bx^{(\theta)}) \circ ((\one - \bp) \wedge \OPT) \rangle &\quad \mbox{(By Claim~\ref{fact2})}\\
	&= \langle \nabla F(\bx^{(\theta)}) \circ (\one - \bx^{(\theta)}), (\one - \bp) \wedge \OPT \rangle 
\end{align*}
In the first inequality, we have used the fact that $(\bx^{(\theta)} \vee \OPT) - (\widehat{\bz} + \bx^{(\theta)}) = (\bx^{(\theta)} + (\one - \bx^{(\theta)}) \circ \OPT) - (\widehat{\bz} + \bx^{(\theta)}) = (\one - \bx^{(\theta)}) \circ ((\one - \bp) \wedge \OPT)$. 

In order to upper bound $\langle \nabla F(\bx^{(\theta)}) \circ (\one - \bx^{(\theta)}), (\one - \bp) \wedge \OPT \rangle$, we consider the following vector.
	\[ \widehat{\bv} \coloneqq (1 - \alpha) ((\one - \bp) \wedge \OPT) + \alpha \bp.\]
The intuition behind the choice of $\widehat{\bv}$ is to connect the fact that the inner product $\langle \nabla F(\bx^{(\tau)}) \circ (\one - \bx^{(\tau)}), \bv \rangle$ is not much more than $\langle \nabla F(\bx^{(\tau)}) \circ (\one - \bx^{(\tau)}), \alpha \OPT \rangle$ to the insight that  $\langle \nabla F(\bx^{(\tau)}) \circ (\one - \bx^{(\tau)}), (\one - \bp) \wedge \OPT \rangle$ is small.

Since $\alpha \ge 1/2$, we have $1 - \alpha \leq \alpha$. Therefore, for each $i \in [n]$, we have
\begin{equation*}
	{\widehat{\bv}}_i =
	\begin{cases}
		(1 - \alpha)(1 - \bp_i) + \alpha \bp_i \leq \alpha &\quad \text{if $i \in \OPT$}\\
		\alpha \bp_i \leq \alpha &\quad \text{otherwise}
	\end{cases}
\end{equation*}
Therefore $\|\widehat{\bv}\|_{\infty} \leq \alpha$. Additionally, $\widehat{\bv} \in C$, since it is a convex combination of two vectors in $C$ (recall that $C$ is downward closed and convex). It follows from the definition of $\bv^{(\theta)}$ on line 5 that
\begin{align*}
	\langle & \nabla F(\bx^{(\theta)}) \circ (\one - \bx^{(\theta)}), \bv^{(\theta)} \rangle\\
  &\geq \langle \nabla F(\bx^{(\theta)}) \circ (\one - \bx^{(\theta)}), \widehat{\bv} \rangle\\
  &= (1 - \alpha) \langle \nabla F(\bx^{(\theta)}) \circ (\one - \bx^{(\theta)}), (\one - \bp) \wedge \OPT) \rangle +  \alpha \langle \nabla F(\bx^{(\theta)}) \circ (\one - \bx^{(\theta)}), \bp \rangle \\
	&\geq (1 - \alpha) \langle \nabla F(\bx^{(\theta)}) \circ (\one - \bx^{(\theta)}), (\one - \bp) \wedge \OPT) \rangle + \alpha \langle \nabla F(\bx^{(\theta)}) \circ (\one - \bx^{(\theta)}), \OPT \rangle\\
	&\geq (1 - \alpha) \langle \nabla F(\bx^{(\theta)}) \circ (\one - \bx^{(\theta)}), (\one - \bp) \wedge \OPT) \rangle + \alpha (F(\bx^{(\theta)} \vee \OPT) - F(\bx^{(\theta)}))
\end{align*}

By rearranging the inequality above, we obtain
\begin{align*}
	& \langle \nabla F(\bx^{(\theta)}) \circ (\one - \bx^{(\theta)}), (\one - \bp) \wedge \OPT) \rangle\\
	&\quad \leq {1 \over 1 - \alpha} \langle \nabla F(\bx^{(\theta)}) \circ (\one - \bx^{(\theta)}), \bv^{(\theta)} \rangle - {\alpha \over 1 - \alpha}	(F(\bx^{(\theta)} \vee \OPT) - F(\bx^{(\theta)})). 
\end{align*}

By combining all of the inequalities, we obtain
\begin{align*}
	F(\widehat{\bz}) &\geq F(\widehat{\bz} + \bx^{(\theta)}) - F(\bx^{(\theta)})\\
	&\geq F(\bx^{(\theta)} \vee \OPT) - F(\bx^{(\theta)}) - \langle \nabla F(\bx^{(\theta)}) \circ (\one - \bx^{(\theta)}), (\one - \bp) \wedge \OPT) \rangle\\
	&\geq {1 \over 1 - \alpha} (F(\bx^{(\theta)} \vee \OPT) - F(\bx^{(\theta)}) - \langle \nabla F(\bx^{(\theta)}) \circ (\one - \bx^{(\theta)}), \bv^{(\theta)} \rangle)\\
	&\geq {1 \over 1 - \alpha} (e^{- \alpha \theta} F(\OPT) - F(\bx^{(\theta)}) - \langle \nabla F(\bx^{(\theta)}) \circ (\one - \bx^{(\theta)}), \bv^{(\theta)} \rangle) 
\end{align*}
In the last inequality, we have used Corollary~\ref{cor:max-x} and Lemma~\ref{lem:x-or-opt} to lower bound $F(\bx^{(\theta)} \vee \OPT)$ by $e^{- \alpha \theta} F(\OPT)$.

Since $\widehat{\bz}$ is a candidate solution for the Double Greedy step on line 14, it follows from Lemma~\ref{lem:box-constraint} that $F(\bz) \geq {1 \over 2} F(\widehat{\bz})$, and the theorem follows.

\subsection{Combining the two solutions}
\label{sec:combine}

As we have shown above, in every iteration $\theta$ of the algorithm, we obtain two solutions $\by^{(1)}$ and $\bz$ satisfying:
\begin{align*}
	F(\by^{(1)}) &\geq  {1 \over e} \left((1 - \theta) e^{(1 - \alpha) \theta} F(\OPT) + e^{\theta} F(\bx^{(\theta)}) \right)\\
	F(\bz) &\geq {1 \over 2(1 - \alpha)} (e^{-\alpha \theta} F(\OPT) - F(\bx^{(\theta)}) - \langle \nabla F(\bx^{(\theta)}) \circ (\one - \bx^{(\theta)}), \bv^{(\theta)} \rangle) 
\end{align*}

In the following, we show that there is an iteration $\theta$ for which $\max\{F(\by^{(1)}), F(\bz)\} > C \cdot F(\OPT)$, where $C$ is a constant that we will set later. We will proceed by contradiction and assume that $\max\{F(\by^{(1)}), F(\bz)\} \leq C \cdot F(\OPT)$ for all $\theta$. 

Note that the coefficient of $F(\bx^{(\theta)})$ is positive in the first inequality above and it is negative in the second inequality, and there is a trade-off between the two solutions. We can get a handle on this trade-off as follows. To simplify matters, we take a convex combination of the two inequalities and eliminate $F(\bx^{(\theta)})$. Thus we get that, for all $\theta \in [0, 1]$,
	\[ {1 \over {1 \over 2(1 - \alpha)} + e^{\theta - 1}} \left( {(2 - \theta) e^{(1 - \alpha) \theta - 1} \over 2(1 - \alpha)} F(\OPT) - {e^{\theta - 1} \over 2(1 - \alpha)} \langle \nabla F(\bx^{(\theta)}) \circ (\one - \bx^{(\theta)}), \bv^{(\theta)} \rangle \right) \leq C \cdot F(\OPT). \]
Thus, for all $\theta \in [0, 1]$,
	\[ \langle \nabla F(\bx^{(\theta)}) \circ (\one - \bx^{(\theta)}), \bv^{(\theta)} \rangle \geq \Big((2 - \theta) e^{-\alpha \theta} - (e^{1 - \theta} + 2(1 - \alpha))C \Big) F(\OPT).\] 
Now note that
\begin{align*}
	F(\bx^{(\theta)}) &= \int_0^{\theta} \langle \nabla F(\bx^{(t)}) \circ (\one - \bx^{(t)}), \bv^{(t)} \rangle) dt\\
	&\geq \int_0^{\theta} \Big((2 - t) e^{-\alpha t} - (e^{1 - t} + 2(1 - \alpha))C \Big) F(\OPT) dt\\
	&= \left({e^{- \alpha t} (\alpha (t - 2) + 1) \over \alpha^2} - (-e^{1 - t} + 2(1 - \alpha)t) C \right) \Big\rvert_{t = 0}^{\theta} \;\; F(\OPT)\\
	&= \left({e^{- \alpha \theta} (\alpha (\theta - 2) + 1) + 2\alpha - 1 \over \alpha^2} - (-e^{1 - \theta} + 2(1 - \alpha) \theta + e) C \right) F(\OPT)
\end{align*}
Therefore
\begin{align*}
	F(\by^{(1)}) \geq \Bigg(& (1 - \theta) e^{(1 - \alpha) \theta - 1} + {e^{(1 - \alpha) \theta - 1} (\alpha (\theta - 2) + 1) + e^{\theta - 1} (2\alpha - 1) \over \alpha^2}\\
	&\quad - (-1 + 2(1 - \alpha) \theta e^{\theta - 1} + e^{\theta}) C \Bigg) F(\OPT)  
\end{align*}

In order to obtain a contradiction, we need that the coefficient of $F(\OPT)$ in the inequality above is at least $C$ for some $\theta \in [0, 1]$ and some $\alpha \in [1/2, 1]$. Equivalently,
\[ \left( (1 - \theta) e^{(1 - \alpha) \theta - 1} + {e^{(1 - \alpha) \theta - 1} (\alpha (\theta - 2) + 1) + e^{\theta - 1} (2\alpha - 1) \over \alpha^2} - (2(1 - \alpha) \theta e^{\theta - 1} + e^{\theta}) C \right) \geq 0.\]
By rearranging, we have
\[ C \leq {1 \over 2(1 - \alpha) \theta e^{\theta - 1} + e^{\theta}} \left( (1 - \theta) e^{(1 - \alpha) \theta - 1} + {e^{(1 - \alpha) \theta - 1} (\alpha (\theta - 2) + 1) + e^{\theta - 1} (2\alpha - 1) \over \alpha^2} \right). \]
Thus, in order to obtain the best approximation $C$, we need to maximize the right hand side of the inequality above over $\theta \in [0, 1]$ and $\alpha \in [1/2, 1]$. Setting $\alpha = 1/2$ and $\theta = 0.18$ gives $C > 0.372$.

\bibliographystyle{plain}
\bibliography{names,submodular}

\newpage
\appendix

\allowdisplaybreaks

\section{Analysis of Algorithm~\ref{alg:double-greedy-box}}
\label{app:double-greedy-box}

In this section, we analyze Algorithm~\ref{alg:double-greedy-box} given in Section~\ref{sec:double-greedy-box}. We start by showing that the problem $\max_{\bu \leq \bx \leq \bv} F(\bx)$ has an optimal solution $\OPT$ with the following property.

\begin{lemma}
\label{lem:opt-box}
	There is an optimal solution $\OPT$ to the problem $\max_{\bu\le\bx\le\bv} F(\bx)$ such that for all $i$, either $\OPT_i = \bu_i$ or $\OPT_i = \bv_i$.
\end{lemma}
\begin{proof}
	Let $\OPT = \argmax_{\bu \leq \bx \leq \bv} F(\bx)$ be an arbitrary optimal solution. Note that we can write each $\OPT_i$ as a convex combination of $\bu_i$ and $\bv_i$: $\OPT_i = \gamma_i \bu_i + (1-\gamma_i) \bv_i$, where $\gamma_i = (\OPT_i - \bu_i) / (\bv_i - \bu_i)$ if $\bu_i \neq \bv_i$ and $\gamma_i = 1$ otherwise. Let $\bo_i$ be a random variable that is equal to $\bu_i$ with probability $\gamma_i$ and is equal to $\bv_i$ with probability $1-\gamma_i$. Let $\bo$ be the vector whose coordinates are $\bo_i$'s. By the definition of the multilinear extension, we have $F(\OPT) = \Ex_{\bo} [F(\bo)]$. Thus, there exists a realization $\bo = \hat{\bo}$ such that $F(\hat{\bo}) \ge F(\OPT)$. Thus there is an optimal solution such that for all $i$, its $i$th coordinate is either $\bu_i$ or $\bv_i$.
\end{proof}

We will also use the following observation that follows from the definition of the multilinear extension.

\begin{claim}
\label{claim1}
  Let $\bx \in [0, 1]^n$ and $\delta \in [- \bx_i, 1 - \bx_i]$. We have
  \[ F(\bx + \delta \cdot \one_{\{i\}}) - F(\bx) = \delta \cdot (F(\bx \vee \one_{\{i \}}) - F(\bx \wedge \one_{V \setminus \{i\}})).\] 
\end{claim}
\begin{proof}
  Note that, for each $\by \in [0, 1]^n$ and each $j \in [n]$,
    \[ F(\by) = \by_j \cdot F(\by \vee \one_{\{j\}}) + (1 - \by_j) \cdot F(\by \wedge \one_{V\setminus \{j\}}).\]
  Therefore
  \begin{align*}
    F(\bx + \delta \cdot \one_{\{i\}}) &= (\bx_i + \delta) \cdot F(\bx \vee \one_{\{i\}}) + (1 - \bx_i - \delta) \cdot F(\bx \wedge \one_{V\setminus \{i\}}),\\
    F(\bx) &= \bx_i \cdot F(\bx \vee \one_{\{i\}}) + (1 - \bx_i) \cdot F(\bx \wedge \one_{V\setminus \{i\}}),
  \end{align*}
  and the claim follows.
\end{proof}

\begin{proofof}{Lemma~\ref{lem:box-constraint}}
	Let $\OPT = \argmax_{\bu \le \bx \le \bv} F(\bx)$. By Lemma~\ref{lem:opt-box}, we may assume that for each $i$, either $\OPT_i = \bu_i$ or $\OPT_i = \bv_i$. Let $\OPT^{(i)} = (\OPT \vee \bu^{(i)})\wedge \bv^{(i)}$. Note that $\bu^{(i)} \le \OPT^{(i)} \le \bv^{(i)}$ and therefore $\nabla F(\bu^{(i)}) \ge \nabla F(\OPT^{(i)}) \ge \nabla F(\bv^{(i)})$ by Fact~\ref{fact2}.

	We will show that
	\begin{equation} \label{eq:eq2}
		F(\OPT^{(i-1)}) - F(\OPT^{(i)}) \le \frac{1}{2}\left(F(\bu^{(i)}) - F(\bu^{(i-1)}) + F(\bv^{(i)}) - F(\bv^{(i-1)})\right).
	\end{equation}
  Note that (\ref{eq:eq2}) immediately implies the lemma. We prove (\ref{eq:eq2}) in the following.

	Suppose that $a'_i + b'_i \neq 0$. By Claim~\ref{claim1},
  \[
		F(\bu^{(i)}) - F(\bu^{(i-1)}) = \frac{a_i'}{a_i'+b_i'} (\bv_i - \bu_i) \( F(\bu^{(i-1)} \vee \one_{\{i\}}) - F(\bu^{(i-1)} \wedge \one_{V \setminus \{i\}})\) = \frac{(a_i')^2}{a_i' + b_i'}.
	\]
	and
	\[
		F(\bv^{(i)}) - F(\bv^{(i-1)}) = \frac{b_i'}{a_i'+b_i'} (\bv_i - \bu_i) \( - F(\bv^{(i-1)} \vee \one_{\{i\}}) + F(\bv^{(i-1)} \wedge \one_{V \setminus \{i\}})\) = \frac{(b_i')^2}{a_i' + b_i'}.
	\]
	Now recall that we have either $\OPT_i = \bu_i$ or $\OPT_i = \bv_i$. If $\OPT_i = \bu_i$, we have
	\begin{align*}
		F(\OPT^{(i-1)}) - F(\OPT^{(i)}) &=\frac{a_i'}{a_i'+b_i'} (\bv_i - \bu_i) \( - F(\OPT^{(i-1)} \vee \one_{\{i\}}) + F(\OPT^{(i-1)} \wedge \one_{(V\setminus\{i\})})\)\\
		&= -\frac{a_i'}{a_i'+b_i'} (\bv_i - \bu_i) \frac{\partial F}{\partial x_i}(\OPT^{(i-1)})\\
		&\le -\frac{a_i'}{a_i'+b_i'} (\bv_i - \bu_i) \frac{\partial F}{\partial x_i}(\bv^{(i-1)})\\
    &= \frac{a_i' b_i}{a_i'+b_i'}\\
		&\le \frac{a_i' b'_i}{a_i' + b_i'}
	\end{align*}
  On the first two lines, we have used Claim~\ref{claim1} and Fact~\ref{fact1}. On the third line, we have used the fact that $\nabla F(\OPT^{(i - 1)}) \geq \nabla F(\bv^{(i - 1)})$. On the fourth and fifth lines, we have used the definition of $b_i$ and $b'_i$.

  If $\OPT_i = \bv_i$, we use an analogous argument.
	\begin{align*}
		F(\OPT^{(i-1)}) - F(\OPT^{(i)}) &= \frac{b_i'}{a_i'+b_i'} (\bv_i - \bu_i) (F(\OPT^{(i-1)} \vee \one_{\{i\}}) - F(\OPT^{(i-1)} \wedge \one_{V\setminus \{i\}}))\\
		&= \frac{b_i'}{a_i'+b_i'} (\bv_i - \bu_i) \frac{\partial F}{\partial x_i}(\OPT^{(i-1)})\\
		&\le \frac{b_i'}{a_i'+b_i'} (\bv_i - \bu_i) \frac{\partial F}{\partial x_i}(\bu^{(i-1)})\\
		&\le \frac{a_i' b_i'}{a_i' + b_i'}
	\end{align*}
  Since $2a'_ib'_i \leq (a'_i)^2 + (b'_i)^2$, the inequality (\ref{eq:eq2}) follows.

  Finally, suppose that $a'_i + b'_i = 0$. Notice that $a_i \ge -b_i$ by submodularity so this case can only happen if $a_i = b_i = 0$. In this case, we can set the $i$th coordinate of $\bu$ and $\bv$ to an arbitrary common value and by the same argument as above, we have $F(\bu^{(i)}) = F(\bu^{(i-1)}), F(\bv^{(i)}) = F(\bv^{(i-1)})$, and $F(\OPT^{(i)}) = F(\OPT^{(i-1)})$. Therefore (\ref{eq:eq2}) is trivially satisfied.
\end{proofof}

\section{Discretized algorithm}
\label{app:discretized}

\begin{figure}[t!]
\centering
\begin{minipage}{0.7\linewidth}
\begin{algorithm}[H]
\label{alg:main-discretized}
\caption{Discretized algorithm for a general constraint}
Initialize $\bx^* = 0$ \;
$\delta = n^{-4}$ \;
\For{$\theta \in \{0, \delta, 2\delta, 3\delta, \ldots, 1\}$}{
  \tcp{Dampened Continuous Greedy stage}
  Initialize $\bx^{(0)} = \zero$ \;
  $t = 0$\;
  \While{$t < \theta$}{ 
    $\bv^{(t)} = \argmax_{\bc \in C, \textcolor{red}{\|\bc\|_{\infty} \le \alpha}} \langle \nabla F(\bx^{(t)}) \circ (\one - \bx^{(t)}), \bc \rangle$ \;
    $\bx^{(t + \delta)} = \bx^{(t)} + \delta \bv^{(t)} \circ (\one - \bx^{(t)})$ \;
    $t = t + \delta$\;
  }
  \;
  \tcp{Standard Continuous Greedy stage}
  Initialize $\by^{(\theta)} = \bx^{(\theta)}$ \;
  \While{$t < 1$}{ 
    $\bv^{(t)} = \argmax_{\bc \in C} \langle \nabla F(\by^{(t)}) \circ (\one - \by^{(t)}), \bc \rangle$ \;
    $\by^{(t + \delta)} = \by^{(t)} + \delta \bv^{(t)} \circ (\one - \by^{(t)})$ \;
    $t = t + \delta$\;
  }
  \;
  \tcp{Double Greedy stage}
  $\bp = \argmax_{\bc \in C} \langle \nabla F(\bx^{(\theta)}) \circ (\one - \bx^{(\theta)}), \bc \rangle$ \;
  Find $\bz$ approximating $\argmax_{\bc \le \bp} F(\bc)$ using the Double Greedy algorithm (Algorithm~\ref{alg:double-greedy-box})\;
  \;
  \tcp{Update the best solution}
  \If{$F(\by^{(1)}) > F(\bx^*)$}{
		$\bx^* = \by^{(1)}$ \;
  }
	\If{$F(\bz) > F(\bx^*)$}{
		$\bx^* = \bz$ \;
	}
}
Return $\bx^*$ \;
\end{algorithm}
\end{minipage}
\caption{Discretized algorithm for a general constraint}
\label{fig:alg-discretized}
\end{figure}

In Figure~\ref{fig:alg-discretized}, we give a discretized version of Algorithm~\ref{alg:main}. In the remainder of this section, we show how to modify the analysis from Section~\ref{sec:analysis}.

We discretize the time interval $[0,1]$ into segments of size $\delta = n^{-4}$.

\paragraph{Analysis of the solution $\by^{(1)}$.}
We modify the analysis from Subsection~\ref{sec:y-value} as follows. (Note that the rest of the analysis remains unchanged.) 

Consider an iteration $\theta$ of the outer for loop of Algorithm~\ref{alg:main-discretized}. In the remainder of this section, all of the vectors $\bx^{(\cdot)}$, $\by^{(\cdot)}$, etc. refer to the vectors during iteration $\theta$. All of the time steps are implicitly assumed to be the discrete time steps $\{0, \delta, 2\delta, 3\delta, \ldots, 1\}$.

First we prove bounds on $\|\bx^{(t)}\|_\infty$ and $\|\by^{(t)}\|_\infty$ with a similar argument to Lemma~\ref{lem:max-norm}. Consider a fixed $i \in [n]$. Since $\bv_i^{(t)} \le \alpha$ for all $t < \theta$, we have $1 - \bx_i^{(t + \delta)} \ge (1 - \delta\alpha)(1 - \bx_i^{(t)})$ for all  $t < \theta$. Thus,
\begin{align*}
	1 - \bx_i^{(t)} &\ge (1 - \delta\alpha)^{t/\delta}.
\end{align*}
Similarly, for all $t \in [\theta, 1]$ and $t$ a multiple of $\delta$, 
	\[ 1 - \by_i^{(t)} \ge (1 - \delta\alpha)^{\theta/\delta}(1-\delta)^{(t - \theta)/\delta}. \]

Next we prove a lower bound on $F(\by^{(1)})$ with a similar argument to Theorem~\ref{thm:y-value}.

\begin{lemma}
	Consider $\bu, \bv$ satisfying $|\bu_i - \bv_i| \le \delta$. Then
		\[ F(\bv) - F(\bu) \le \delta n^2 M, \]
	where $M = \max_{i \in [n]} f(\{i\})$.
\end{lemma}
\begin{proof}
	First consider the case $\bu, \bv$ agree on all but coordinate $i_0$. Let $R(\bx)$ be the random set where each element $i$ is independently included with probability $\bx_i$. This process can be thought of as picking a random $r_i \in [0,1]$ and including $i$ if $r_i \le \bx_i$. One can generate (coupled) $R(\bu)$ and $R(\bv)$ by sharing the same $r_i$'s. Notice that $R(\bu)$ and $R(\bv)$ agree on all coordinates other than $i_0$ and they disagree on coordinate $i_0$ with probability at most $\delta$. Thus we have 
	\[ \Ex[f(R(\bu)) - f(R(\bv))] \le \delta \Ex[f(R(\bu)) \;|\; R(\bu) \ne R(\bv)] \le \delta \max_{S \subseteq V} f(S) \le \delta n M. \]
	Next, let $\bu^i$ be the vector whose first $i$ coordinates agree with $\bu$ and the last $n-i$ coordinates agree with $\bv$. We have
	\[ F(\bv) - F(\bu) = \sum_{i=0}^{n-1} (F(\bu^i) - F(\bu^{i+1})) \le \delta n^2 M, \]
	where the inequality comes from the above special case.
\end{proof}
\begin{corollary}
	Consider $\bu$ and $\bv$ satisfying $\bu \leq \bv$ and $\bv_i \leq \bu_i + \delta$. Then
		\[ \frac{\partial F}{\partial x_i}(\bu) - \frac{\partial F}{\partial x_i}(\bv) \le 2\delta n^2 M. \]
\end{corollary}
\begin{proof}
	By Claim~\ref{fact1},
		\[ \frac{\partial F}{\partial x_i}(\bu) - \frac{\partial F}{\partial x_i}(\bv) = F(\bu \vee \one_{\{i\}}) - F(\bv\vee \one_{\{i\}}) - F(\bu \wedge \one_{V\setminus\{i\}}) + F(\bv \wedge \one_{V\setminus\{i\}})\le 2\delta n^2 M. \]
\end{proof}

For any given $t \in [\theta, 1 - \delta]$, we have
\begin{align*}
	& F(\by^{(t + \delta)}) - F(\by^{(t)})\\
	&= \int_{0}^{1} \langle \nabla F((1-z) \by^{(t)} + z \by^{(t+\delta)}), \by^{(t+\delta)} - \by^{(t)} \rangle dz\\
	&= \langle \nabla F(\by^{(t)}), \by^{(t+\delta)} - \by^{(t)} \rangle + \int_{0}^{1} \langle \nabla F((1-z) \by^{(t)} + z \by^{(t+\delta)}) - \nabla F(\by^{(t)}), \by^{(t+\delta)} - \by^{(t)} \rangle dz\\
	&\ge \langle \nabla F(\by^{(t)}), \by^{(t+\delta)} - \by^{(t)} \rangle - \int_{0}^{1} \|\nabla F((1-z)\by^{(t)} + z\by^{(t+\delta)}) - \nabla F(\by^{(t)})\|_\infty \|\by^{(t+\delta)} - \by^{(t)} \|_1 dz\\
	&\ge \langle \nabla F(\by^{(t)}), \by^{(t+\delta)} - \by^{(t)} \rangle - \delta^2 n^3 M\\
	&= \langle \nabla F(\by^{(t)}), \delta \bv^{(t)} \circ (\one - \by^{(t)})\rangle - \delta^2 n^3 M\\
	&= \langle \nabla F(\by^{(t)}) \circ (\one - \by^{(t)}), \delta \bv^{(t)} \rangle - \delta^2 n^3 M\\
	&\geq \langle \nabla F(\by^{(t)}) \circ (\one - \by^{(t)}), \delta \OPT \rangle - \delta^2 n^3 M\\
	&\geq \delta (F(\by^{(t)} \vee \OPT) - F(\by^{(t)})) - \delta^2 n^3 M\\
	&\geq \delta (1-\delta\alpha)^{\theta/\delta}(1-\delta)^{(t-\theta)/\delta} \cdot F(\OPT) - \delta\cdot F(\by^{(t)}) - \delta^2 n^3 M \quad \mbox{(By Lemma~\ref{lem:x-or-opt})}
\end{align*}

Therefore we have the following recurrence for $F(\by^{(t)})$:
\[ F(\by^{(t + \delta)}) \geq (1 - \delta) F(\by^{(t)}) + \delta (1-\delta\alpha)^{\theta/\delta}(1-\delta)^{(t-\theta)/\delta} \cdot F(\OPT) - \delta^2 n^3 M.\] 
By expanding the recurrence, we obtain
\[ F(\by^{(t)}) \geq (1 - \delta)^{(t - \theta)/\delta} \left((t - \theta) (1 - \delta \alpha)^{\theta/\delta} \cdot F(\OPT) + F(\bx^{(\theta)})\right) - t\delta n^3 M. \]

By substituting $t=1$ and $\delta=n^{-4}$ and by using the inequalities $e^{-x/(1-x)} \le 1-x \le e^{-x}$ for all $0<x<1$, we obtain
\begin{align*}
	F(\by^{(1)}) &\geq e^{-1 + \theta} \left((1 - \theta) e^{-\alpha \theta} F(\OPT) + F(\bx^{(\theta)}) \right) - O(n^{-1}) F(\OPT)\\
	&= {1 \over e} \left((1 - \theta) e^{(1 -\alpha) \theta} F(\OPT) + e^{\theta} F(\bx^{(\theta)}) \right) - O(n^{-1}) F(\OPT).
\end{align*}

\paragraph{Combining the two solutions $\by^{(1)}$ and $\bz$.}
We extend the argument in Subsection~\ref{sec:combine} as follows. As we have shown above, in every iteration $\theta$ that is a multiple of $\delta$, we obtain two solutions $\by^{(1)}$ and $\bz$ satisfying:
\begin{align*}
	F(\by^{(1)}) &\geq  {1 \over e} \left((1 - \theta) e^{(1 - \alpha) \theta} F(\OPT) + e^{\theta} F(\bx^{(\theta)}) \right) - O(n^{-1}) F(\OPT)\\
	F(\bz) &\geq {1 \over 2(1 - \alpha)} (e^{-\alpha \theta} F(\OPT) - F(\bx^{(\theta)}) - \langle \nabla F(\bx^{(\theta)}) \circ (\one - \bx^{(\theta)}), \bv^{(\theta)} \rangle) 
\end{align*}

As before, we show that there is an iteration $\theta$ for which $\max\{F(\by^{(1)}), F(\bz)\} > C \cdot F(\OPT)$, where $C$ is a constant that we will set later. We will proceed by contradiction and assume that $\max\{F(\by^{(1)}), F(\bz)\} \leq C \cdot F(\OPT)$ for all $\theta$. 

We take a convex combination of the two inequalities and eliminate $F(\bx^{(\theta)})$. Thus we get that, for all $\theta$,
	\[ {1 \over {1 \over 2(1 - \alpha)} + e^{\theta - 1}} \left( {(2 - \theta) e^{(1 - \alpha) \theta - 1} \over 2(1 - \alpha)} F(\OPT) - {e^{\theta - 1} \over 2(1 - \alpha)} \langle \nabla F(\bx^{(\theta)}) \circ (\one - \bx^{(\theta)}), \bv^{(\theta)} \rangle) \right) \leq C' \cdot F(\OPT), \]
where $C' = C + O(n^{-1})$. Thus, for all $\theta \in [0, 1]$,
	\[ \langle \nabla F(\bx^{(\theta)}) \circ (\one - \bx^{(\theta)}), \bv^{(\theta)} \rangle) \geq \Big((2 - \theta) e^{-\alpha \theta} - (e^{1 - \theta} + 2(1 - \alpha))C' \Big) F(\OPT).\] 
For any given $t \in [0, \theta)$, we have
\begin{align*}
	& F(\bx^{(t + \delta)}) - F(\bx^{(t)}) \\
	&= \int_{0}^{1} \langle \nabla F((1-z) \bx^{(t)} + z \bx^{(t+\delta)}), \bx^{(t+\delta)} - \bx^{(t)} \rangle dz\\
	&= \langle \nabla F(\bx^{(t)}), \bx^{(t+\delta)} - \bx^{(t)} \rangle + \int_{0}^{1} \langle \nabla F((1-z) \bx^{(t)} + z \bx^{(t+\delta)}) - \nabla F(\bx^{(t)}), \bx^{(t+\delta)} - \bx^{(t)} \rangle dz\\
	&\ge \langle \nabla F(\bx^{(t)}), \bx^{(t+\delta)} - \bx^{(t)} \rangle - \int_{0}^{1} \|\nabla F((1-z)\bx^{(t)} + z\bx^{(t+\delta)}) - \nabla F(\bx^{(t)})\|_\infty \|\bx^{(t+\delta)} - \bx^{(t)} \|_1 dz\\
	&\ge \langle \nabla F(\bx^{(t)}), \bx^{(t+\delta)} - \bx^{(t)} \rangle - \delta^2 n^3 M\\
	&= \langle \nabla F(\bx^{(t)}), \delta \bv^{(t)} \circ (\one - \bx^{(t)})\rangle - \delta^2 n^3 M\\
	&= \delta \langle \nabla F(\bx^{(t)}) \circ (\one - \bx^{(t)}), \delta \bv^{(t)} \rangle - \delta^2 n^3 M\\
	&\geq \delta \Big((2 - t) e^{-\alpha t} - (e^{1 - t} + 2(1 - \alpha))C' \Big) F(\OPT) - \delta^2 n^3 M
\end{align*}
By expanding the recurrence, we obtain
\begin{equation}\label{eq:fx}
	F(\bx^{(t + \delta)}) \geq \sum_{i = 0}^{t/\delta} \left( \delta \Big((2 - t + i \delta) e^{-\alpha (t - i \delta)} - (e^{1 - t + i\delta} + 2(1 - \alpha))C' \Big) F(\OPT) - \delta^2 n^3 M \right)
\end{equation}

By computation, we have
\begin{align*}
\sum_{i = 0}^{t/\delta} \delta (2 - t + i \delta) e^{-\alpha (t - i \delta)}&=\delta e^{-\alpha t} \left((2-t-\delta)\frac{e^{\alpha\delta(t/\delta+1)}-1}{e^{\alpha\delta}-1} + \delta\frac{(t/\delta+1)e^{\alpha\delta(t/\delta+2)}-(t/\delta+2)e^{\alpha\delta(t/\delta+1)}+1}{(e^{\alpha\delta}-1)^2}\right)\\
&\ge e^{-\alpha t} \left(\frac{(2-t)(e^{\alpha t}-1)}{\alpha} + \frac{t\alpha e^{\alpha t} -e^{\alpha t} + 1}{\alpha^2}\right) - O(\delta)\\
&= \frac{(t-2)\alpha e^{-\alpha t} + 2\alpha -1 + e^{-\alpha t}}{\alpha^2} - O(\delta)
\end{align*}
We also have
\begin{align*}
\sum_{i = 0}^{t/\delta} \delta(e^{1 - t + i\delta} + 2(1 - \alpha))C' &= \left(\delta e^{1-t} \frac{e^{t+\delta}-1}{e^{\delta}-1} + 2(1-\alpha)(t+\delta)\right)C'\\
&\le \left(e - e^{1-t}+ 2(1-\alpha)(t+\delta)\right)C' + O(\delta) 
\end{align*}

Substituting into Equation~\ref{eq:fx}, we have
\begin{equation*}
F(\bx^{(t + \delta)}) \geq \left(\frac{(t-2)\alpha e^{-\alpha t} + 2\alpha -1 + e^{-\alpha t}}{\alpha^2} - (e - e^{1-t}+ 2(1-\alpha)t)C' - O(\delta n^3)\right)F(\OPT)
\end{equation*}

Notice that the RHS is exactly the same as in the continuous argument except for an additive error $O(\delta n^3)F(\OPT)$. By carrying this error through, we can finish the proof exactly the same as before.
\end{document}